\documentclass[]{article}

\usepackage{fancyhdr}

\input{body/macros}

\begin{document}
\thispagestyle{empty}
	
\title{Robust Real-time Computing with Chemical Reaction Networks
  \thanks{This research supported in part by NSF grants 1900716 and 1545028.}}

\date{}

\author{
	\small{Willem Fletcher$^1$,
	Titus H. Klinge$^{2 [0000-0002-2297-6712]}$,
	James I. Lathrop$^3$,} \\
	\small{Dawn A. Nye$^{3 [0000-0003-1192-2740]}$,
	and Matthew Rayman$^3$} \\
	\\
	\small{$^1$ Carleton College, Northfield, MN, USA, fletcherw@carleton.edu} \\
	\small{$^2$ Drake University, Des Moines, IA, USA, titus.klinge@drake.edu} \\
	\small{$^3$ Iowa State University, Ames, IA, USA, jil,omacron,marayman@iastate.edu}
}

\maketitle

\begin{abstract}
 Recent research into analog computing has introduced new notions of computing real numbers.
 Huang, Klinge, Lathrop, Li, and Lutz defined a notion of computing real numbers in real-time with chemical reaction networks (CRNs), introducing the classes $\lcrn$ (the class of all Lyapunov CRN-computable real numbers) and $\rtcrn$ (the class of all real-time CRN-computable numbers).
 In their paper, they show the inclusion of the real algebraic numbers $ALG \subseteq \lcrn \subseteq \rtcrn$ and that $ALG \subsetneqq \rtcrn$ but leave open where the inclusion is proper.
 In this paper, we resolve this open problem and show \( \alg = \lcrn \subsetneqq \rtcrn \).
 However, their definition of real-time computation is fragile in the sense that it is sensitive to perturbations in initial conditions.
 To resolve this flaw, we further require a CRN to withstand these perturbations.
 In doing so, we arrive at a discrete model of memory.
 This approach has several benefits.
 First, a bounded CRN may compute values approximately in finite time.
 Second, a CRN can tolerate small perturbations of its species' concentrations.
 Third, taking a measurement of a CRN's state only requires precision proportional to the exactness of these approximations.
 Lastly, if a CRN requires only finite memory, this model and Turing machines are equivalent under real-time simulations.
\end{abstract}

\section{Introduction}
\label{sec:introduction}

Over the last few decades, many theories of molecular computing have emerged.
These theories help inform experimental research and help explore the boundaries of nanoscale computation.
Some models of molecular programming are structural, such as algorithmic self-assembly~\cite{FURCY2021,Hader2021}; some models are amorphous, such as chemical reaction networks~\cite{CAPPELLETTI202064,Severson2020}; and some models combine these to characterize more complex interactions~\cite{clamons2020,klinge_et_al:LIPIcs:2020:12959}.
Since molecular programming is a relatively new field, many open problems exist concerning the computational limits of these models.

Investigating the complexity of computing real numbers in computational models has historically significant roots.
In Turing's famous 1936 paper~\cite{turing1936}, he defined a real number to be computable if its ``expression as a decimal is calculable with finite means.''
Real numbers can also be classified according to how efficiently they can be computed by a Turing machine.
For example, rational numbers are efficiently computable because their recurring decimal pattern can be produced in real time---even by a finite automaton.
More formally, a number \( \alpha\in\mathbb{R} \) is \emph{real-time computable by a Turing machine} if \( n \) bits of its fractional component can be produced in \( O(n) \) time.
Many transcendental numbers are known to be real-time computable, but surprisingly, no irrational algebraic number is known to be real-time computable.
In fact, in 1965, Hartmanis and Stearns conjectured that if \( \alpha\in\mathbb{R} \) is real-time computable by a Turing machine, then it is either rational or transcendental~\cite{HartmanisStearns1965}.

Recent research into analog computing introduced new notions of computing real numbers.
Bournez \etal introduced the notion of computing a real number \emph{in the limit} with a general purpose analog computer (GPAC)~\cite{Bournez2006}.
To compute \( \alpha\in\mathbb{R} \) ``in the limit,'' a designated variable \( x(t) \) must satisfy \( \lim_{t\rightarrow\infty}x(t) = \alpha \).
Computing real numbers in this way has also been investigated in population protocols~\cite{Bournez2012} and chemical reaction networks (CRNs)~\cite{HuangGPAC2019}.
Huang \etal defined a number \( \alpha\in\mathbb{R} \) to be \emph{real-time computable by chemical reaction networks}, written \( \alpha\in\rtcrn \), if there exists a CRN with integral rate constants and a designated species \( X \) such that, if all species concentrations are initialized to zero, then \( x(t) \) converges to \( \alpha \) exponentially quickly~\cite{Huang2019}.
This means that after \( n \) seconds, the concentration of \( X \) is within \( 2^{-n} \) of \( \alpha \), so the CRN gains one bit of accuracy every second.
Huang \etal also required that all species concentrations be bounded to avoid the so-called \emph{Zeno paradox} of performing an infinite amount of computation in finite time using a fast-growing catalyst species~\cite{Pouly2017}.
When this restriction is lifted, the measure of time is no longer linear but rather a function of arc length.
In this sense, no power is lost via imposing a boundedness requirement.
Further, it eliminates the undesirable Zeno paradox from the model.

A key aspect of Huang \etal's definition of \( \rtcrn \) is the requirement that the CRN be initialized to \emph{all zeros}, prohibiting any encoding of \( \alpha \) in the initial condition of the CRN.
The authors showed that \( e, \pi \in \rtcrn \), leveraging the fact that the initial condition is \emph{exact}.
However, these constructions fail if their initial conditions are perturbed by any \( \epsilon > 0 \).
Huang \etal also defined a subfield of \( \rtcrn \) they called \emph{Lyapunov CRN-computable real numbers}, written \( \lcrn \).
The definition of \( \lcrn \) is similar to \( \rtcrn \) except with the additional constraint that the terminating state of the CRN must be an \emph{exponentially stable equilibrium point}.
Since an exponentially stable equilibrium point is \emph{attracting}, any initial condition within its basin of attraction will converge exponentially quickly to it.
As a result, any \( \alpha\in\lcrn \) can be computed even in the presence of bounded perturbations to initial conditions.
Huang \etal also proved that \( \alg\subseteq\lcrn\subseteq\rtcrn \) where \( \alg \) is the set of algebraic real numbers.
The authors left as an open problem which of these inclusions is strict.

An additional consequence of computing a real number \( \alpha \) ``in the limit'' with CRNs is that recovering the bits of \( \alpha \) is difficult.
Even if we produce \( \alpha \) \emph{exactly} in the concentration of a species \( X \), we cannot read its individual bits without an infinitely precise measurement device.
Alternatively, if a CRN produced the bits of \( \alpha \) as a sequence of measurable memory states, then the bits can be read even with imperfect measurements.

Another limitation of this method of computation is in implementation.
The concentration of a species in a solution containing a CRN is ultimately determined by the discrete, integral count of the species.
This places a countable limit on the number of ``exact'' values a concentration can achieve even when a CRN is otherwise perfectly initialized and executed.
In the mass action model, we often wave away this issue precisely because we do not have an infinitely precise measurement device.
This does, however, somewhat obviate the point of being able to calculate values precisely.
In fact, previous results concerning CRNs frequently abuse this hand waving to reach theorems that are true of the mass action kinematics but not of the reality it models.
Instead, a more reasonable question to ask is what values can we calculate robustly, quickly, and approximately.

In this paper, we show \( \alg = \lcrn \subsetneqq \rtcrn \) to resolve the open problem stated above.
This fully characterizes what values we may compute robustly and quickly; however, this definition of computation yet suffers from the inherent flaws described above.
To resolve this weakness of the model, we acknowledge these limitations and loosen the definition of computation to accept approximate results.
To do so, we only require that a CRN produces approximations of numbers in the sense that an open interval around a concentration $\alpha$ is in an equivalence class with $\alpha$ itself.
This approach has three major benefits.
First, a bounded CRN may compute not only a single value in finite time but also a sequence of values.
Second, a CRN can tolerate small perturbations of its species' concentrations (and potentially other parameters).
Third, taking a measurement of a CRN's state only requires precision proportional to the smallest of these intervals.

If we then fix a collection of these intervals into collection of memory maps for a CRN's species and allow it to compute their corresponding memory states in sequence, we obtain a discrete model characterizing a robust chemical computer.
Indeed, given, in this sense, a robust CRN and a memory map which fully describes it, a Turing machine may simulate the CRN by maintaining a tape for each of its species indicating what memory state that species is in.
We show that this simulation can be done in real-time for CRNs which use only a finite amount of memory.
Although we conjecture that CRNs which use an unbounded amount of memory can also be simulated in real-time (which, if true, would unify the analog and discrete Hartmanis-Stearns Conjectures), finite memory suffices for many real world applications.

The rest of the paper is organized as follows.
Section~2 reviews some necessary preliminaries used in the remainder of the paper.
Section~3 resolves the open problem \( \alg = \lcrn \subsetneqq \rtcrn \).
Section~4 characterizes CRNs in terms of a robust memory map.
Lastly, Section~5 discusses the consequences of the proceeding sections.

\section{Preliminaries}
\label{sec:preliminaries}

A \emph{Chemical Reaction Network} (CRN), $N$, is a tuple $N = (S,R)$, where $S$ is a finite number of species and $R$ is a finite set of reactions on those species.
In this paper we investigate \emph{deterministic} CRNs, \ie, CRNs under deterministic mass action semantics that are modeled with systems of differential equations~\cite{oEpsPoj98}.
Given a deterministic CRN, let $x_i(t)$ denote the real-valued concentration of $X_i$ at time $t$ for each species $X_i \in S$.
Let $\textbf{x} = (x_1,\ldots,x_n)$ denote the state of $N$, where $n = \card{S}$.
We write the rate of change of each $x_i$ as $\odx{x_i}{t} = f_i(x_1,\ldots,x_n)$ and the rate of change of the entire system as $\odx{\textbf{x}}{t} = \textbf{f}_N = (f_1,\ldots,f_n)$. 
Each $f_i$ is a polynomial determined by $N$~\cite{oEpsPoj98}.
In this paper, rate constants for each reaction in $R$ are integral, and thus each $f_i \in \integers[x_1,\ldots,x_n]$.
Furthermore, the initial concentrations of the species, given by an initial state $\textbf{x}(0) = \textbf{x}_0$, along with $\textbf{f}_N$ determine the unique behavior of $N$. Lastly, when $f_N(\textbf{z}) = 0$, we call $\textbf{z}$ a fixed point.

The definition of real-time computable by a CRN used in this paper is given by~\cite{HuangGPAC2019,Huang2019}.  We repeat the definition here for convenience.
\begin{definition}
	A real number $\alpha$ is \emph{real-time computable by CRNs} if there exists a CRN $N = (S,R)$ and a species $X \in S$ with the following properties:
	\begin{enumerate}
		\item (\emph{Integrality.}) All rate constants of $R$ are positive integers.
		\item (\emph{Boundedness.}) The concentration $x_i(t)$ for each species in $S$ is bounded by a constant $\beta$ for all time $t \in [0, \infty)$ when $\textbf{x}_0 = 0$.
		\item (\emph{Real-Time Convergence.}) If $N$ is initialized with $\textbf{x}_0 = 0$, then for all times $t \ge 1$, $\abs*{x(t) - \abs*{\alpha}} < 2^{-t}$.
	\end{enumerate}
	We denote the set of all real-time CRN-computable real numbers as $\rtcrn$.
\end{definition}
Excluding the species that converges to $\alpha$, the above definition places no restrictions on any species beyond that they be bounded.
In many cases, this may be undesirable.
The next definition formalizes the notion of converging to a single state, at which point the CRN can be considered finished. 
\begin{definition}
	An \emph{exponentially stable point} of a CRN is a state $\textbf{z} \in \nnr^n$ for which there exists $\alpha, \delta, C > 0$ such that, if the CRN is initialized to a state $\textbf{x}_0$ satisfying $\abs*{\textbf{z} - \textbf{x}_0} < \delta$, then for all times $t \ge 0$, $\abs*{\textbf{z} - \textbf{x}(t)} \le C e^{-\alpha t} \abs*{\textbf{z} - \textbf{x}_0}$.
\end{definition}

\begin{definition}
	\label{Lyapunov CRN Computability}
	A real number $\alpha$ is \emph{Lyapunov-CRN computable} if there exists a CRN $N = (S,R)$, a species $X_i \in S$, and a state $\textbf{z}$ with $\textbf{z}(X_i) = \abs*{\alpha}$ that satisfies the following properties:
	\begin{enumerate}
		\item (\emph{Integrality.}) All rate constants of $R$ are positive integers.
		\item (\emph{Boundedness.}) The concentration $x_i(t)$ for each species in $S$ is bounded by a constant $\beta$ for all time $t \in [0, \infty)$ when $\textbf{x}_0 = 0$.
		\item (\emph{Exponential Stability.}) $\textbf{z}$ is an exponentially stable point.
		\item (\emph{Convergence.}) If $N$ is initialized with $\textbf{x}_0 = 0$, then $\underset{t \rightarrow \infty}{\lim} \textbf{x}(t) = \textbf{z}$.
	\end{enumerate}
	We denote the set of all Lyapunov-CRN computable real numbers as $\reals_\text{LCRN}$.
\end{definition}

\begin{observation}
	\label{lem:Exponentially Stable Points Are Fixed Points}
	If $\textbf{z}$ is an exponentially stable point of a CRN, then it is a fixed point of that CRN.
\end{observation}
Note that the converse of Observation~\ref{lem:Exponentially Stable Points Are Fixed Points} is not true.

We use $\alg$ to denote the set of real algebraic numbers of the rationals.
This is the set of real numbers which are the root of some polynomial $f \in \rationals[x]$ with rational coefficients.

\section{Lyapunov Reals are Algebraic}
\label{sec:LyapunovAlgebraic}

To investigate robustness issues in real-time computing, we first look at the relationship between $\lcrn$ and $\alg$ and show that $\lcrn = \alg$.
As a consequence, a bounded CRN may only compute the algebraic numbers reliably in the sense that they exist inside of a potential well.
Since Huang \etal proved that $\alg \subsetneqq \rtcrn$ and $\alg \subseteq \lcrn \subseteq \rtcrn$~\cite{Huang2019}, it suffices to show that $\alg = \lcrn$ to resolve that $\lcrn \subsetneqq \rtcrn$.
We prove this result in two parts.
First, we show that every exponentially stable fixed point is isolated.
Second, we show that isolated fixed points necessarily have algebraic components.

Let $E_N$ denote the set of exponentially stable points of a CRN, $N$, and let $F_N$ denote the set of fixed points of $N$.
Recall that fixed points are not necessarily isolated (consider a CRN which does nothing once initialized), however, the set of exponentially stable fixed points, $E_N$, \emph{are} isolated in $F_N$ (not just $E_N$), which we prove below.

\begin{lemma}
	\label{Exponentially Stable Points Are Isolated}
	If $N$ is a CRN and $\bm{z} \in E_N$, then $\bm{z}$ is isolated in $F_N$.
\end{lemma}

\begin{proof}
	Suppose that $\bm{z}$ is not an isolated element of $F_N$.
	
	Let $\alpha, \delta, C > 0$ be such that they satisfy $\bm{z}$'s exponential stability.
	Since $\bm{z}$ is not isolated, there is a distinct $\bm{x_0} \in F_N$ within $\delta$ of $\bm{z}$ with $0 < \abs*{\bm{z} - \bm{x_0}} < \delta$.
	By definition, $\bm{x_0}$ is a fixed point, hence $\bm{f}_N(\bm{x_0}) = 0$.
	It then follows that if $N$ is initialized to $\bm{x_0}$, we have $\bm{x}(t) = \bm{x_0}$ for all times $t$.
	Hence,
	\[ \abs*{\bm{z} - \bm{x_0}} = \abs*{\bm{z} - \bm{x}(t)} \le C e^{-\alpha t} \abs*{\bm{z} - \bm{x_0}}. \]
	Since the left-hand side of this equation is a positive constant while the right-hand goes to $0$ as $t$ goes to infinity, no such $\bm{x_0}$ can exist and, consequentially, $\bm{z}$ must be isolated.
	\qed
\end{proof}

In order to show that isolated fixed points (and by Lemma~\ref{Exponentially Stable Points Are Isolated}, each exponentially stable point) can only have algebraic components, we utilize a powerful property of real closed fields.

\begin{definition}
	A field $F$ is a real closed field if and only if it satisfies either of the following equivalent definitions~\cite{basu_algorithms_2006}:
	\begin{itemize}
		\item $F$ satisfies the sentence $\sigma$ (written $\models_F \sigma$) if and only if $\models_{\reals} \sigma$.
		\item $F$ is not algebraically closed, but $F[\sqrt{-1}]$ is.
	\end{itemize}
\end{definition}

\begin{lemma}
	$\alg$ is a real closed field.~\cite{basu_algorithms_2006}
\end{lemma}

In short, since $\alg$ is a real closed field, we need only show that there is a first-order sentence that fully captures what it means to be an isolated fixed point of a CRN. This then implies that only values in $\alg$ can satisfy it.

\begin{lemma}
	\label{Isolated  Points Are Algebraic}
	If $\bm{z}$ is a fixed point of a CRN, $N$, that is isolated in $F_N$, then the components of $\bm{z}$ are in \alg.
\end{lemma}

\begin{proof}
	Assume the hypothesis, let $N$ be a CRN, and let $\bm{z}$ be an exponentially stable point of $N$.
	
	To prove the lemma, we will construct a first order logic sentence and invoke a useful property stemming from real closed fields.
	For some neighborhood $\mathscr{U}_{\bm{z}}$ around $\bm{z}$, we will construct the following sentence, $\sigma_{\bm{z}}$.
	\[
	\sigma_{\bm{z}} \equiv
	\exists \bm{x} \paren*{
		\bm{f}_N \paren*{\bm{x}} = 0
		\land \bm{x} \in \mathscr{U}_{\bm{z}}
		\land \paren*{
			\exists \bm{y} \paren*{\bm{f}_N \paren*{\bm{y}} = 0
				\land \bm{y} \in \mathscr{U}_{\bm{z}}} \rightarrow \bm{x} = \bm{y}}}
	\]
	Intuitively, $\sigma_{\bm{z}}$ states that there is an open ball around $\bm{z}$ such that there are no other elements of $F_N$ inside of it.
	This is, of course, the definition of being an isolated point.
	
	To begin, we require a language for our first-order logic.
	It suffices to use the language of rings (and fields), $L = (0,1,+,\cdot,-)$, where these symbols all have their usual meaning.
	It's routine to verify that the associated axioms are all easily written in first order logic.
	
	Now that we have a language, we construct a structure $\mathfrak{U}$ with the usual required properties.
	\begin{itemize}
		\item Assign to $\forall$ a universe, $\mathscr{U}$.
		\item Assign to each $n$-place predicate symbol a subset of $\mathscr{U}^n$.
		\item Assign $\func[f^{\mathfrak{U}}]{\mathscr{U}^n}{\mathscr{U}}$ for each $n$-place function.
		\item Assign to each constant an element of $\mathscr{U}$.
	\end{itemize}
	We choose a universe $\mathscr{U}$ such that it forms a field with the binary operations $+$ and $\cdot$.
	We will explicitly specify it later.
	For now, we rewrite the terms of $\sigma_{\bm{z}}$ in order to work within the constraints of the tools available to first-order logic.
	\begin{itemize}
		\item Vectors are an undefined syntactic sugar in first-order logic.
		We instead write $\exists \bm{x}$ as $\exists x_1 \ldots \exists x_n$.
		\item With only $0$ and $1$ available, we write an integer $n$ as the sum of $n$ $1$s, the additive inverse of such a sum, or as simply $0$.
		For a rational, $v = \frac{a}{b}$, we write $\exists v (v \cdot b = a)$.
		\item We write each $f_i$ in its full form since $f_i \in \integers[x_1,\ldots,x_n]$.
		Similarly, we write $\bm{f}_N = 0$ as $f_1 = 0 \land \ldots \land f_n = 0$.
		\item We write $a < b$ as $\exists c (b = c^2 + a )\land a \ne b$.
		\item Lastly, we must rewrite $\bm{x} \in \mathscr{U}_{\bm{z}}$.
		Since $\bm{z}$ is an exponentially stable point, it is an isolated point by Lemma \ref{Exponentially Stable Points Are Isolated}.
		As such, there is an open ball of radius $r > 0$ centered at $\bm{z}$ containing no other fixed points.
		Let $\epsilon = \frac{r}{\sqrt{n}}$.
		For each coordinate $z_i$ of $\bm{z}$, pick rational numbers $l_i \in (z_i - \epsilon,z_i)$ and $r_i \in (z_i,z_i + \epsilon)$.
		Then we now write $l_1 < x_1 \land x_1 < r_1 \land \ldots \land l_n < x_n \land x_n < r_n$ as a valid neighborhood of $\bm{z}$ for $\bm{x} \in \mathscr{U}_{\bm{z}}$.
		Upon inspection, it's clear that the maximum distance from $\bm{z}$ to any point within this neighborhood is strictly less than $r$.
	\end{itemize}
	Thus we can write $\sigma_{\bm{z}}$ in terms of first order logic, and by construction $\models_{\reals} \sigma_{\bm{z}}$.
	
	Since $\alg$ is a real closed field, it follows that $\models_{\alg} \sigma_{\bm{z}}$.
	As a consequence, there is a fixed point, $\bm{w}$, in $\mathscr{U}_{\bm{z}}$ with all coordinates in $\alg$.
	However, since $\bm{z}$ and $\bm{w}$ are both fixed points in $\mathscr{U}_{\bm{z}}$ and $\models_{\reals} \sigma_{\bm{z}}$, it must be the case that $\bm{z} = \bm{w}$, so $\bm{z}$ has coordinates all in $\alg$.
	\qed
\end{proof}

Using these lemmas, it is now straightforward to prove the theorem.

\begin{theorem}
	\label{thm:algeqlup}
	$\alg = \lcrn$
\end{theorem}

\begin{proof}
	Let $\alpha \in \lcrn$, and let $N$, $X_i$, and $\bm{z}$ be the CRN, designated species, and exponentially stable point that testify to this.
	By definition, $\bm{z}$ is exponentially stable;
	by Lemma~\ref{Exponentially Stable Points Are Isolated}, $\bm{z}$ is isolated in $F_N$;
	by Lemma~\ref{Isolated  Points Are Algebraic}, every component of $\bm{z}$ is algebraic.
	Thus, $\bm{z}(X_i) = |\alpha|$ is algebraic, and therefore $\alpha \in \alg$.
	As the converse is known~\cite{Huang2019}, we have $\alg = \lcrn$.
	\qed
\end{proof}
\section{A Robust Notion of Memory in CRNs}
\label{sec:mmcrn}

In the previous section, we concerned ourselves with CRNs which are permitted infinite precision to compute real values robustly in the limit.
This excuses several impossibilities for the elegance of its model at the expense of realism.
In practice, these CRNs would compute their intended values robustly in approximation and would require only finite time.

In this section we explore the consequences of requiring a CRN to be robust in this sense, that is that they compute values approximately in finite time.
In particular, we characterize the behavior of these robust CRNs in terms of these approximations to arrive at a somewhat paradoxical discrete model of analog computing.

Recall that boundedness is one of the three criteria for a real-time CRN.
For this section, we use the following definitions of boundedness.

\begin{definition}
	A CRN $N = (S,R)$ is \emph{$\beta$-bounded} at $\bm{x_0} \in \nnr^S$ if, when initialized to $\bm{x_0}$, there exists some $\beta > 0$ such that $x < \beta$ for each $X \in S$. Moreover, $N$ is \emph{uniformly $\beta$-bounded} on $O \subseteq  \nnr^S$ if there is some $\beta > 0$ for which $N$ is bounded on each $\bm{x_0} \in O$ by $\beta$.
\end{definition}

Unless otherwise specified, a bounded CRN is initialized to the point at which it is bounded.
Similarly, a uniformly bounded CRN is initialized to a point at which it is bounded (and is implicitly bounded at any initial point).

There are two natural ways by which a CRN may compute a number $\alpha$.
It may either do so exactly when a species' concentration becomes $\alpha$ or in the limit as per Lyapunov-CRN computability, real-time computability, or some slower manner.
Both approaches, however, are imperfect.
In the latter case, the concentration of the species computing $\alpha$ either must always maintain a non-zero distance from $\alpha$ after any finite time or, at best, suffers from the same limitation of computing $\alpha$ exactly: the inability to remain at $\alpha$.
The following theorem and corollary formalize this notion.

\begin{theorem}
	\label{Rest Is for the Weak}
	Let $N = (S,R)$ be a bounded CRN.
	For each species $X \in S$, $x$ is either constant or the set of times for which $\frac{dx}{dt} = 0$ is countable.
\end{theorem}

\begin{proof}
	Let $N$ be as given and fix $X \in S$.
	Assume that $x$ is non-constant but that $\text{ker}(x')$ is uncountably infinite.
	As such, $\text{ker}(x')$ must have infinitely many accumulation points in $\nnr$.
	In particular, it contains an sequence $\collection{a_n}{n}{\naturals}$ not containing $0$ such that $\underset{n \rightarrow \infty}{\lim} a_n \neq 0$.
	Then since $x'$ is analytic on the connected open set $\positives$~\cite{oKraPar02}, by the Identity Theorem~\cite{lewis_global_analysis_2014}, $x' = 0$ on $\positives$.
	Moreover, since $x'$ is continuous on $\nnr$, it must be the case that $x'(0) = 0$.
	But if $x'$ is identically $0$, then $x$ must be constant, a contradiction.
	\qed
\end{proof}

\begin{corollary}
	\label{No Exact Numbers}
	Let $N = (S,R)$ be a bounded CRN. Pick $c \in \nnr$.
	Then for any non-constant species $X \in S$, the set of times $t \in \nnr$ where $x(t) = c$ is countable.
\end{corollary}

\begin{proof}
	For $X \in S$, suppose $x$ is not constant.
	When $x$ attains the value $c$, either $x' = 0$ or $x' \neq 0$.
	In the former case, $\text{ker}(x')$ is countable by Theorem~\ref{Rest Is for the Weak}.
	In the latter case, let $t_1 < t_2$ be such that $x(t_1) = x(t_2) = c$.
	Since $x$ is differentiable, there is some $t_3 \in (t_1,t_2)$ such that $x'(t_3) = 0$ by the mean value theorem.
	But if between any two times where $x = c$ there is another time where $x' = 0$, then there can at most be a countable number of such times since $\text{ker}(x')$ is countable.
	\qed
\end{proof}

It is clear from Corollary~\ref{No Exact Numbers} that computing an exact value with a CRN is, if not impossible, then a less meaningful concept than one would prefer.
This is not inherently problematic as a model of computation.
A CRN is capable of computing any computable function in the limit~\cite{Fages2017}.

In each of these models, however, there is the implicit assumption that a CRN may be precisely constructed by which we mean each rate constant and the initial concentration of each species is exactly as prescribed.
In practice, this is impractical, which leads us to a notion of robustness.
A CRN, informally speaking, is ``robust'' if it can tolerate a small perturbation of its concentrations (or rate constants) at any time without affecting its function.
This is intuitively a difficult task since changing any such condition clearly alters the solution to the system of ODEs describing the CRN.

Exponentially stable points are a good example of robustness in the following sense.
If a CRN manages to get within an $\epsilon$-ball of such a point $\bm{z}$, it proceeds to $\bm{z}$ in the limit without exception.
Ideally, a robust CRN would transition from exponentially stable point to exponentially stable point during its computation with some outside force periodically driving it away from each stable equilibrium.

Exponential stability is a far stricter requirement than is necessary to compute a number $\alpha$, but it does illustrate an important point.
If a CRN computes $\alpha$ either in the limit or for longer than a countable set of times, there is always a buffer zone around it which must necessarily be considered in an equivalence class with $\alpha$.
In Figure~\ref{fig:runningExample}, this corresponds to the intervals labeled $A$, $B$, and $C$ which could be considered equivalence classes for $1$, $\frac{1}{2}$, and $0$ respectively.
We formalize this notion in the following theorems and definitions.

\begin{theorem}
	\label{Buffer Zone}
	Let $N = (S,R)$ be a bounded CRN, and let $X \in S$ be a non-constant species.
	For any time $t_0 \in \nnr$, there exists a $\delta > 0$ such that for all $t \in (t_0,t_0 + \delta)$, $x(t) \neq x(t_0)$.
	Moreover, there exists an $\epsilon > 0$ and a $t \in (t_0,t_0 + \delta)$ such that $\abs*{x(t) - x(t_0)} > \epsilon$.
\end{theorem}

In order to prove this theorem, we first prove a supporting lemma.

\begin{lemma}
	\label{Well-Ordered Zeroes}
	Let $N = (S,R)$ be a bounded CRN.
	Then for any non-constant species $X \in S$, the set $\text{ker}(x')$ is well-ordered.
\end{lemma}

\begin{proof}
	Suppose that $N$ is a bounded CRN and let $X \in S$ be a non-constant species such that $\text{ker}(x')$ is not well-ordered.
	Then there must be a decreasing sequence of times $\collection{t_n}{n}{\naturals}$ in $\text{ker}(x')$ (all bounded below by $0$) which converge to an accumulation point $t \in \nnr$.
	Assume without loss of generality that $t > 0$.
	Then since $x'$ is analytic on the connected open set $\positives$, $x' = 0$ on $\positives$ by the Identity Theorem.
	Finally, because $x'$ is continuous, it follows that $x'(0) = 0$.
	But $x'$ is non-constant, a contradiction.
	\qed
\end{proof}

\begin{proof}[of Theorem~\ref{Buffer Zone}]
	Let $N$ be as given and let $X \in S$ be a non-constant species.
	Fix a time $t_0 \in \nnr$.
	If there is no time $t_1 > t_0$ for which $x(t_1) = x(t_0)$, then pick $\delta = 1$ and $\epsilon = \frac{1}{2} \abs*{x(t_0) - x(t_0 + 1)}$ to conclude the proof.
	
	Otherwise, let $t_1 > t_0$ be any time for which $x(t_1) = x(t_0)$.
	By the mean value theorem, there must be a time $t \in (t_0,t_1)$ for which $x'(t) = 0$.
	But by Lemma~\ref{Well-Ordered Zeroes}, there is a least element of $\text{ker}(x')$ greater than $t_0$.
	Let that element be $t$.
	Picking $\delta = t - t_0$ and $\epsilon = \frac{1}{2} \abs*{x(t_0) - x(t)}$ concludes the proof.
	\qed
\end{proof}

\begin{figure}[t]
	\begin{center}
	\includegraphics[width=4.5in]{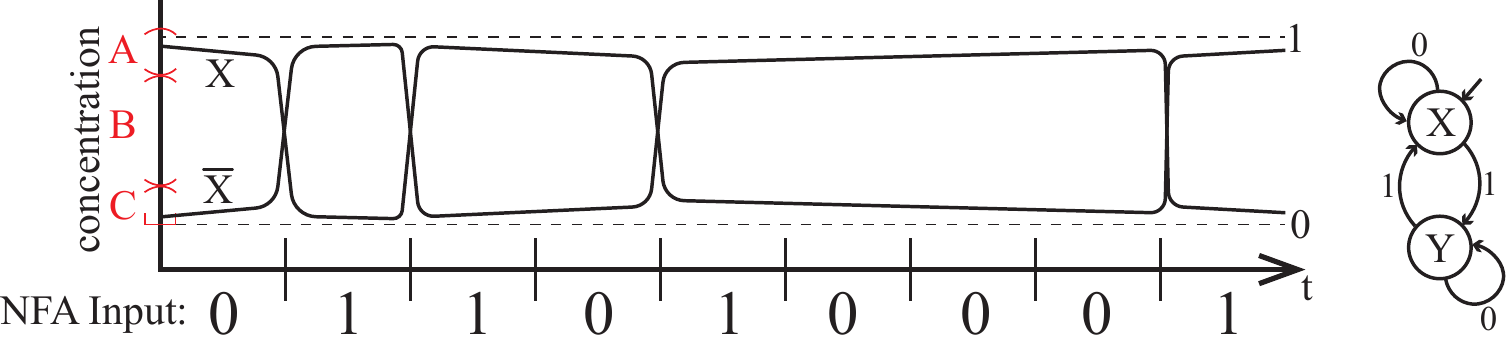}
	\caption{\label{fig:runningExample}
		In~\cite{KLINGE2020114}, Klinge, Lathrop, and Lutz provide a general CRN construction for nondeterministic finite automata (NFAs).
		These NFAs utilize a dual rail system for each state $Z$ with $z(t) \approx 1$ indicating that the NFA is in state $Z$ at time $t$ and $z(t) \approx 0$ indicating the NFA is not in state $Z$ while the complementary species $\overline{Z}$ has the opposite meaning.
		Above, we graph the concentration of $X$ and $\overline{X}$ as input changes and show the approximation regions.}
	\end{center}
\end{figure}

In light of Theorem~\ref{Buffer Zone}, we state a notion of computation useful (but alone insufficient) for CRNs.

\begin{definition}
	A CRN $N = (S,R)$ \emph{$(\epsilon,d)$-computes} a real number $\alpha \in \nnr$ if there is an $X \in S$ and a time $t_0 \in \nnr$ such that $\abs*{x(t) - \alpha} < \epsilon$ for all $t \in (t_0,t_0 + d)$.
\end{definition}

Less formally, a CRN \emph{$(\epsilon,d)$-computes} a real number $\alpha$ if it gets close enough to it for a long enough time.
To continue our earlier example, the correct choice of $\epsilon$ and $d$ make $x(t)$ correctly compute $0$ and $1$ but never the garbage state $\frac{1}{2}$ in Figure~\ref{fig:runningExample}.
This underscores that the particular choice of these two parameters is critical for the CRN's intended purpose.
Indeed, a species $X$ of a \emph{bounded} CRN so computes every element of the closure of its image for some single choice of $d$ for every $\epsilon$ and vice versa!
The latter is obvious (pick $\epsilon$ to be larger than the CRN's bound), and we formally state the former.

\begin{theorem}
	\label{Epsilon-D Computation Needs Stricter Requirements}
	Let $N = (S,R)$ be a bounded CRN, and let $\epsilon > 0$.
	Then there exists a $d > 0$ such that each $X \in S$ $(\epsilon,d)$-computes every element of $cl(x(\nnr))$.
\end{theorem}

\begin{proof}
	Let $N$ and $\epsilon > 0$ be given.
	Fix $X \in S$.
	Since $N$ is bounded, there exists $\beta > 0$ such that $\abs*{x'} < \beta$.
	Define $d = \frac{\epsilon}{2 \beta}$.
	
	Now let $\alpha \in x(\nnr)$.
	By definition, there exists a time $t \in \nnr$ such that $x(t) = \alpha$.
	Then since $\abs*{x'} < \beta$, it must be the case that $\abs*{x(t) - x(t_0)} = \abs*{x(t_0) - \alpha} < \beta * d = \frac{1}{2} \epsilon < \epsilon$ for every $t_0 \in (t,t + d)$.
	Thus $N$ $(\epsilon,d)$-computes $\alpha$.
	
	It remains to show that $N$ $(\epsilon,d)$-computes the elements of $cl(x(\nnr)) \setminus x(\nnr)$.
	Since $x$ is continuous, there can be at most two such elements, $\inf(x(\nnr))$ and $\sup(x(\nnr))$.
	It suffices to show $N$ $(\epsilon,d)$-computes $l = \inf(x(\nnr))$ as the proof is identical for the supremum.
	
	If $l \in x(\nnr)$, there's nothing left to prove, so assume otherwise.
	Then since $l$ is an accumulation point of $x(\nnr)$, there is some time $t \in \nnr$ for which $\abs*{x(t) - l} < \frac{1}{2} \epsilon$.
	Moreover, it follows that $\abs*{x(t_0) - l} = \abs*{x(t_0) - x(t) + x(t) - l} < \beta * d + \frac{1}{2} \epsilon = \epsilon$ for every $t_0 \in (t,t + d)$.
	Thus $N$ $(\epsilon,d)$-computes $l$.
	\qed
\end{proof}

The following definition resolves this $\epsilon,d$ conundrum described above by eliminating any overlap of $(\epsilon,d)$-computed real numbers.

\begin{definition}
	A CRN $N = (S,R)$ \emph{unambiguously computes} a set $A \subseteq \nnr$ if for each $\alpha \in A$ there exists a species $X \in S$ which $(\epsilon_{\alpha},d_{\alpha})$-computes $\alpha$ for some $\epsilon_{\alpha},d_{\alpha} > 0$ and for each distinct $\alpha_1,\alpha_2 \in A$ which $X$ $(\epsilon_{\alpha_1},d_{\alpha_1})$-computes and $(\epsilon_{\alpha_2},d_{\alpha_2})$-computes respectively, the intervals $(\alpha_1 - \epsilon_{\alpha_1},\alpha_1 + \epsilon_{\alpha_1})$ and $(\alpha_2 - \epsilon_{\alpha_2},\alpha_2 + \epsilon_{\alpha_2})$ are disjoint.
\end{definition}

This notion of unambiguous computation leads directly to a robust notion of CRN memory, but we first state a motivating theorem behind its construction.

\begin{theorem}
	\label{CRNs Dislike Dense Subsets}
	No CRN can unambiguously compute a somewhere dense subset $D$ of $\nnr$ for any choice of $\epsilon_{\alpha},d_{\alpha} > 0$ for each $\alpha \in D$.
\end{theorem}

\begin{proof}
	Let $N = (S,R)$ be a CRN.
	Suppose that $N$ $(\epsilon,d)$-computes a somewhere dense set $D \subseteq \nnr$.
	Let $O \subsetneqq \nnr$ be an open set in which $D$ is dense, and define $D_O = D \intersect O$.
	Since $S$ is a finite set, we may assume without loss of generality that for each $\alpha \in D_O$, $X \in S$ $(\epsilon_{\alpha},d_{\alpha})$-computes $\alpha$ for some $\epsilon_{\alpha},d_{\alpha} > 0$.
	
	Now fix $\alpha_1 \in D_O$.
	Since $D_O$ is dense in $D$, $D_O$ contains another point $\alpha_2 \in (\alpha_1 - \epsilon_{\alpha_1},\alpha_1 + \epsilon_{\alpha_1})$.
	But the intersection of $(\alpha_1 - \epsilon_{\alpha_1},\alpha_1 + \epsilon_{\alpha_1})$ and $(\alpha_2 - \epsilon_{\alpha_2},\alpha_2 + \epsilon_{\alpha_2})$ is clearly not empty, a contradiction.
	\qed
\end{proof}

Theorem~\ref{CRNs Dislike Dense Subsets} shows that many natural encodings of countably infinite sets to bounded intervals cannot be unambiguously computed by a CRN. An example of such is given below where we encode $1 \rightarrow .1$, $2 \rightarrow .01$, $3 \rightarrow .11$, and so on.

\begin{corollary}
	\label{The Natural Encoding Is Dense}
	Let $\func[f]{\naturals}{\interval*{0,1}}$ be the map
	\[ f(n) = \sum_{i = 0}^{\infty} \paren*{\floor*{\frac{n}{2^i}} \mod 2} 2^{-i - 1}. \]
	No CRN can unambiguously compute the set $f(\naturals)$.
\end{corollary}

\begin{proof}
	First, $f$ clearly converges on all of its domain.
	Then by Theorem~\ref{CRNs Dislike Dense Subsets}, it suffices to show that $f(\naturals)$ is dense in $\interval*{0,1}$.
	
	Fix $a \in \interval*{0,1}$.
	If $a \in f(\naturals)$, there's nothing to prove.
	If $a = 1$, then the sequence $\collection{f(2^n - 1)}{n}{\naturals}$ converges to $a$.
	If $a \notin f(\naturals)$, let $0.a_0a_1a_2\ldots$ be the base $2$ expansion of $a$.
	Define $\func[g]{\naturals}{\naturals}$ by
	\[ g(N) = \sum_{n = 0}^{N - 1} a_n * 2^n. \]
	By construction, the sequence $\collection{f(g(n))}{n}{\naturals}$ converges to $a$.
	Finally, since every $a \in \interval*{0,1}$ is either in $f(\naturals)$ or is the accumulation point of a sequence in $f(\naturals)$, it follows that $f(\naturals)$ is dense in $\interval*{0,1}$.
	\qed
\end{proof}

To avoid this problem, any encoding requires an open interval around each value $\alpha$ the CRN must compute wherein the entire interval is considered to be $\alpha$.
Moreover, a CRN can only have countably many such disjoint sets.
In our running example, Figure~\ref{fig:runningExample} demonstrates three such intervals for each state species.
This leads to the following definition wherein we encode a collection of disjoint open intervals to map to identifying natural numbers.

\begin{definition}
	Let $c \in \positives$. A \emph{memory map} is a map $\func[f]{\extendednaturals}{\powerset{[0,c]}}$ satisfying the following conditions:
	\begin{itemize}
		\parskip0pt \parindent0pt
		\item $f(0) = [0,b)$, where $0 < b \le c$.
		\item $\forall n \in \pints$, $f(n) = (a,b)$, where $a,b \in \rationals$ and $0 < a \le b \le c$.
		\item $\forall m,n \in \naturals$, if $m \neq n$, then $f(n)$ and $f(m)$ are disjoint.
		\item $f(\infty) = [0,c] \setminus \underset{n \in \naturals}{\union} f(n)$ and is countable.
	\end{itemize}
\end{definition}

\begin{definition}
	The set of all memory maps on $[0,c]$ is $\mathcal{M}_c$.
	The \emph{order} of $f \in \mathcal{M}_c$, written $ord(f)$, is the cardinality of the support of $f$ over $\naturals$.
\end{definition}

\begin{definition}
	The \emph{inverse memory map} of $f \in \mathcal{M}_c$ is a map $\func[f^{\leftarrow}]{[0,c]}{\extendednaturals}$ such that for all $r \in [0,c]$, $r \in f(f^{\leftarrow}(r))$.
\end{definition}

In principle, a CRN cannot reasonably be initialized to any state more precise than to an interval of a memory map.
Indeed, the consequence of Corollary~\ref{No Exact Numbers} is the well known fact that if a species ever has a non-zero concentration, it will at almost every time $t > 0$, so no power is gained from being able to initialize a species to $0$.

Before proceeding, the definition of a memory map, it should be noted, is \emph{descriptive} of a CRN, not \emph{prescriptive}.
Any memory map can model any CRN, but not all memory maps model any particular CRN \emph{well}.
For example, any $\beta$-bounded CRN can be modeled by the uninteresting memory map that maps every concentration less than $\beta$ to $0$.
Similarly, a memory map with randomly chosen intervals is both equally valid and equally ill-suited.
We do not yet, however, have all of the definitions necessary to describe what makes for a \emph{good} choice of memory map and so return to this topic later in this section.

Now equipped with a notion of memory, we must define the trajectory of a species $X$ through that memory (and a CRN's trajectory in terms of its species').
This is not inherently clear because a species $X$ must pass over all intermediate memory locations when transitioning between two non-adjacent states.
Even if there are only finitely many such intermediary states, including them in the trajectory provides no additional information.
That $X$ passes through them during the transition is a direct consequence of $x$ being continuous.
In Figure~\ref{fig:runningExample}, for example, we never want to include the $B$ interval in our trajectory.

But since each memory state consists of an open interval, $X$ must spend a non-zero length of time inside of it.
This brings us back to the definition of $(\epsilon,d)$-computability.
If we require $X$ to $(\epsilon,d)$-compute the midpoint of the interval of a memory state with $\epsilon$ being half of the interval's width and $d$ being an adjustable parameter, we can arrive at a useful definition of trajectory.
To fully formalize this, however, we first have to develop a bit more notation.

\begin{definition}
	Let $N = (S,R)$ be a $\beta$-bounded CRN.
	For each $X \in S$, let $f_X \in \mathcal{M}_\beta$ be a memory map.
	A species $X \in S$ is in the \emph{memory state} $m \in \extendednaturals$ at time $t$ if $x(t) \in f_X(m)$.
	Similarly, $N$ is in the memory state $\bm{m} \in \extendednaturals^S$ at time $t$ if for each $X \in S$, $X$ is in the memory state $\bm{m}(X)$. 
\end{definition}

As an abuse of notation, a memory state may also be used as the subset of $\nnr$ ($\nnr^S$) it represents.

\begin{definition}
	Let $N$ be a $\beta$-bounded CRN.
	For $X \in S$, let $f_X \in \mathcal{M}_\beta$ be a memory map.
	$X$ \emph{enters} a memory state $n \in \extendednaturals$ at time $t_0$ if there exists an $\epsilon > 0$ such that for all $0 < \epsilon' < \epsilon$, $x(t_0 - \epsilon') \notin f_X(n)$ and $x(t_0 + \epsilon') \in f_X(n)$.
	Similarly, $X$ \emph{leaves} $n$ at time $t_1$ if there exists an $\epsilon > 0$ such that for all $0 < \epsilon' < \epsilon$, $x(t_1 - \epsilon') \in f_X(n)$ and $x(t_1 + \epsilon') \notin f_X(n)$.
	The \emph{state time} of $X$ in $n$ (with no intermediate states) is the difference $t_1 - t_0$.
\end{definition}

There are a few consequences to the above definitions worth mentioning.
When a species is initialized to a memory state $n$, it leaves $n$ without first having entered $n$.
A species may also transition from $n \in \naturals$ to $\infty$ (or in other words, it may touch the boundary of $n$) and then return to $n$, in which case it does not leave or enter $n$.
This is a desirable property as $\infty$ is not a useful memory state except, perhaps, in the limit as $t \rightarrow \infty$.
Further, no species may enter or leave the memory state $\infty$ by definition.

There remain a few edge cases in the above definitions.
If a species never entered a memory state $n$ before it leaves $n$ (i.e. it was initialized to $n$), then we say it entered at time $t = 0$.
In a similar vein, if a species never leaves a memory state, we say it leaves at $t = \infty$ purely as a matter of notational convenience (even if in the limit it transitions to the memory state $\infty$).

Lastly, we remark that sojourn time (arc length) is generally a better measure of runtime for CRNs~\cite{Pouly2017}.
In the case of bounded CRNs, however, state time suffices as it is always within a constant factor of sojourn time.
This is because each species concentration of a bounded CRN necessarily has a bounded rate of change.
We prove this formally below.

\begin{definition}
	Let $N$ be a $\beta$-bounded CRN.
	For $X \in S$, let $f_X \in \mathcal{M}_\beta$ be a memory map.
	The \emph{sojourn time} of a species $X$ in a memory state $n \in \naturals$ is the arc length of $x$ from between when it enters $n$ at time $t_0$ and when it leaves $n$ at time $t_1$,
	\[ \int_{t_0}^{t_1} \sqrt{1 + \paren*{x' \paren*{t}}^2} dt. \]
	The \emph{state time} of $X$ in $n$ is simply $t_1 - t_0$.
\end{definition}

\begin{lemma}
	\label{Sojourn Time Is Basically State Time}
	Let $N$ be a $\beta$-bounded CRN.
	For $X \in S$, let $f_X \in \mathcal{M}_\beta$ be a memory map.
	Then there exists $\beta' > 0$ such that for every memory state $n$ which $X$ enters, the sojourn time $t_a$ and state time $t_r$ of $X$ in $n$ satisfy $t_a \le \beta' t_r$.
\end{lemma}

\begin{proof}
	Let $N,X,f_X,n$ be as given.
	Since $N$ is bounded, there exists a bound $\beta_0 > 0$ such that $\abs*{x'} < \beta_0$.
	Define $\beta' = \beta_0 + 1$, and let $t_0 < t_1$ be the time $X$ enters state $n$ and leaves state $n$ respectively.
	If $t_1 = \infty$, then $t_a \int_{t_0}^{t_1} \sqrt{1 + \paren*{x' \paren*{t}}^2} dt = \infty = t_1 - t_0 \le \beta' t_r$.
	Otherwise
	\[ t_a = \int_{t_0}^{t_1} \sqrt{1 + \paren*{x' \paren*{t}}^2} dt < \int_{t_0}^{t_1} \sqrt{1 + \paren*{\beta_0}^2} dt < \beta' (t_1 - t_0) = \beta' t_r. \]
	\qed
\end{proof}

We now have the tools necessary to define the trajectory of a CRN.
We first give an informal description here with example and then rigorously define it (see Definition~\ref{Trajectory Definition}).
The trajectory of a $\beta$-bounded CRN $N$ initialized to $\bm{x_0}$ is the ordered sequence of memory states obtained as follows.
Start from the initial memory state $\bm{n_0}$.
Each time one or more species enters a new memory state for which its state time is at least $d$, append the new state of $N$ to the sequence.
Continue indefinitely or until there are no further memory state changes.

This construction avoids the undesirability of recording in-between memory states of other species as they transition to their next memory state.
It also has the added benefit of bringing into the trajectory a notion of a species staying in a memory state for a long time.
In a species trajectory, we merely record where the species's concentration goes to but not for how long it stays there.
In the memory trajectory of a full CRN, however, if a species's memory state only rarely changes, we can see that behavior in how infrequently it changes in comparison to other species.
For example, if one wishes to record a species's memory state at regular intervals, the simple solution is to set up a clock with an appropriate period which has no interaction with the rest of the CRN except to place itself into the memory trajectory as a timestamp.

To complete our running example, Figure~\ref{fig:runningExample} has the following trajectory (for the species $X$, $\overline{X}$): $(A,C)(C,A)(A,C)(C,A)(A,C)$.
In this case, because the construction of the CRN requires that $x + \overline{x} = 1$, a symmetric choice of intervals $A$ and $C$ across $\frac{1}{2}$ causes $X$ and $\overline{X}$ to always be in opposite states at all times.

Using this intuition, we now formally construct the definition of trajectory as follows.

\begin{definition}
	\label{Trajectory Definition}
	Let $N = (S,R)$ be a $\beta$-bounded CRN at $\bm{x_0} \in \nnr^S$.
	For $X \in S$, let $f_X \in \mathcal{M}_\beta$ be a memory map.
	The \emph{memory trajectory} of $N$ when $N$ is initialized to $\bm{x_0}$ with delay $d \in \positives$, written $\mbox{\textbf{traj}}(\bm{x_0},d)$, is a sequence of $\naturals^S$ defined by $\mbox{\textbf{traj}}(\bm{x_0},d)(n)(X) = m_{\mbox{last}}(X,\bm{x_0},\bm{m}_{\mbox{next}}^n(\bm{x_0},0,d),d)$ where $m_{\mbox{last}}$ and $\bm{m}_{\mbox{next}}^n$ are helper functions defined below.
\end{definition}

To define the helper functions in the above definition, let $N = (S,R)$ be a $\beta$-bounded CRN at $\bm{x_0} \in \nnr^S$.
For $X \in S$, let $f_X \in \mathcal{M}_\beta$ be a memory map.
Let $T(X,\bm{x_0},d)$ be the set of times when $X$ enters a memory state for which its state time is at least $d \in \positives$ when $N$ is initialized to $\bm{x_0}$, and define $\bm{T}(\bm{x_0},d) = \underset{X \in S}{\union} T(X,\bm{x_0},d)$.

Next define $\func[\bm{m}_{\mbox{next}}]{\nnr^S \times \nnr \times \positives}{\nnr}$ to be the function which, given an initial state $\bm{x_0}$ of $N$, a time $t \in \nnr$, and a delay $d > 0$, selects the least $t_0 \in \bm{T}(\bm{x_0},d)$ for which $t_0 > t$.

First, by Corollary~\ref{No Exact Numbers}, $f^{\leftarrow}_X = \infty$ only when $X$ is instantaneously between memory states or if $X$ is constant and initialized to such a value.
Both are undesirable system behavior easily avoided.
$\bm{m}_{\mbox{next}}$ specifically excludes the former from trajectories while the latter is a mere matter of initialization.
Unless otherwise specified, we \emph{never} initialize a CRN to such a state even if it is a state for which the CRN is bounded.
Second, there is always a least element of each $\bm{T}(\bm{x_0},d)$ for $\bm{m}_{\mbox{next}}$ to select since a species $X$ can be in at most two memory states (leaving one for the other) per every $d$ interval of time, which we formally state below.

\begin{lemma}
	\label{CRNs Change State Only So Quickly}
	Let $N = (S,R)$ be a $\beta$-bounded CRN initialized to $\bm{x_0} \in \nnr^S$.
	Fix a delay $d \in \positives$.
	Then for any $d$ interval of time, $N$'s memory trajectory contains at most $2 \card{S}$ memory states.
\end{lemma}

\begin{proof}
	Let $N$ and $d$ be as given, and fix a time $t \in \nnr$.
	By definition, the set of times when $X \in S$ enters a memory state, $T(X,\bm{x_0},d)$, has the following property.
	For any two distinct times $t_1,t_2 \in T(X,\bm{x_0},d)$, $\abs*{t_1 - t_2} \ge d$ regardless of the choice of memory map.
	It must then be the case that the union of all sets, $\bm{T}(\bm{x_0},d)$, can at most contain $2 \card{S}$ elements of the interval $\interval*{t,t + d}$.
	As such, the memory trajectory of $N$ can only have at most $2 \card{S}$ memory states during this time.
	\qed
\end{proof}

From this lemma, $\bm{m}_{\mbox{next}}$ is well defined.
We extend its definition to a recursive form as follows.
\begin{displaymath}
	\bm{m}_{\mbox{next}}^n(\bm{x_0},t,d) = \left\lbrace
	\begin{array}{l|r}
		\bm{m}_{\mbox{next}}^{n - 1}(\bm{x_0},\bm{m}_{\mbox{next}}(\bm{x_0},t,d),d) & n > 0 \\
		t & n = 0
	\end{array}
	\right.
\end{displaymath}

To finish formalizing the definition of memory trajectory, $\func[m_{\mbox{last}}]{S \times \nnr^S \times \nnr \times \positives}{\naturals}$ is the function which, given a species $X \in S$, an initial state $\bm{x_0}$ of $N$, a time $t \in \nnr$, and a delay $d > 0$, returns the last memory state $n \in \extendednaturals$ for which species $X$ enters $n$ at a time $t_0 \le t$ when $N$ is initialized to $\bm{x_0}$.
More intuitively, $m_{\mbox{last}}$ remembers the current memory state of a species while it is transitioning to another memory state.

We can at last now state what makes for a good memory map.
The guiding principle behind the choice of memory map is that a CRN in a memory state should behave identically going forward regardless of what particular concentration each species has inside of it.
In the spirit of Theorem~\ref{CRNs Dislike Dense Subsets}, this then leads to the following natural definition.

\begin{definition}
	Let $N = (S,R)$ be a uniformly $\beta$-bounded CRN and, for each $X \in S$, a memory map $f_X \in \mathcal{M}_\beta$.
	Fix a delay $d \in \positives$.
	$N$ is \emph{memory deterministic} (with respect to $\collection{f_X}{X}{S}$ and delay $d$) if there is a function $\func[\delta]{\naturals^S}{\naturals^S}$ such that if $\bm{m} \in \naturals^S$ is a memory state in $N$'s memory trajectory, then the next memory state in $N$'s memory trajectory (if one exists) is $\delta(\bm{m})$. When no such memory state exists, $\delta(\bm{m}) = \bm{m}$.
\end{definition}

It is important to note that the transition function $\delta$ described in the above definition is relative to which memory state(s) the CRN it describes may be initialized.
The behavior of unreachable memory states are outside the scope of the definition.
With respect to any such memory state, $\delta$'s behavior is unrestricted.
In general, $\delta$ itself need not necessarily even be computable (although it typically should be).
For a well behaved CRN (one which admits many possible initializations), however, $\delta$ must satisfy the definition for every valid initialization simultaneously.

Before moving on, observe that \emph{all} bounded CRNs are memory deterministic for \emph{a} choice of memory map.
Recall that the memory map which maps all concentrations to $0$ is valid for every CRN.
Similarly, each CRN modeled with this memory map is memory deterministic with $\delta(\bm{m}) = \bm{m}$.
Otherwise put, a good choice of memory map for a CRN requires not just that it be memory deterministic but that it is also sufficiently refined to produce a useful model.

Now, unsurprisingly, the notion of a memory map bears a strong resemblance to a Turing machine tape.
We know that CRNs and Turing machines are equivalent models~\cite{Fages2017}.
The question remains, however, if one model can outperform the other in some significant way.
We can address this question in one direction by providing a means for a Turing machine to simulate a CRN.
In general, this is a difficult task.
With memory maps, this becomes easier.
Since neither model is allowed to speed up indefinitely, we may treat a single step of a Turing machine as a constant length of time.
Then we may define that a Turing machine simulates a CRN $N$ if it follows $N$'s memory trajectory on its tape(s).
Formally, we have the following definitions.

\begin{definition}
	Fix $k \in \pints$.
	Let $\Lambda = \collection{(\bm{\omega_n},t_n)}{n}{\naturals}$ be a sequence of tuples in $\paren*{\bin^*}^k \times \nnr$.
	A Turing machine with at least $k$ tapes initialized to $\bm{\omega_0}$ \emph{follows} $\Lambda$ if there is a strictly increasing computable sequence $\collection{s_n}{n}{\naturals}$ of $\naturals$ such that for each $i \in \mathbb{Z}_k$, the contents of tape $T_i$ at step $s_n$ is $\bm{\omega_n}(i)$.
	Similarly, $M$ \emph{real-time follows} $\Lambda$ if there is a constant $c > 0$ such that each $s_n \le c t_n$
\end{definition}

\begin{definition}
	Let $N = (S,R)$ be a (uniformly) $\beta$-bounded CRN and, for each $X \in S$, let $f_X \in \mathcal{M}_\beta$.
	Fix a delay $d \in \positives$.
	A Turing machine $M$ \emph{follows $N$ according to $\collection{f_X}{X}{S}$ with delay $d$} if for each $X \in S$ there exists a computable injective map $\func[\mbox{itoa}_X]{\naturals}{\bin^*}$ such that $M$ follows
	\[ \Lambda = \collection{\paren*{\mbox{\textbf{itoa}} \paren*{\mbox{\textbf{traj}} \paren*{\bm{x_0},d}(n)},\bm{m}^n_{\mbox{next}}(\bm{x_0},0,d)}}{n}{\naturals} \]
	when initialized to $\bm{x_0} \in \nnr^S$ (for every initialization $\bm{x_0} \in \nnr^S$ for which $N$ is bounded), where $\mbox{\textbf{itoa}} \paren*{\mbox{\textbf{traj}} \paren*{\bm{x_0},d}(n)}(X) = \mbox{itoa}_X(\mbox{\textbf{traj}} \paren*{\bm{x_0},d}(n)(X))$ for $X \in S$ and $n \in \naturals$.
	Similarly, $M$ \emph{real-time follows $N$ according to $\collection{f_X}{X}{S}$ with delay $d$} if $M$ real-time follows $\Lambda$ when initialized to $\bm{x_0} \in \nnr^S$ (for every initialization $\bm{x_0} \in \nnr^S$ for which $N$ is bounded).
\end{definition}

We can extend our running example to these definitions as follows.
For the itoa functions, interval $A$ maps to $1$, $B$ maps to $10$, and $C$ maps to $0$.
Moreover, since the CRN was constructed directly from a finite automaton, it only takes two steps to compute each subsequent memory state and write it to the appropriate tape.
If follows trivially, then, that there is a Turing machine which real-time follows the CRN.

More generally, since Turing machines and CRNs are equivalent models~\cite{Fages2017}, there is always a Turing machine that follows any CRN $N$.
The more interesting (and far more difficult) question is if there always exists a Turing machine $M$ and some choice of itoa functions for which $M$ real-time follows $N$.
Intuitively, analog computing \emph{should} be more efficient than discrete computing in some respect.
Indeed, were a CRN either unbounded or if it were allowed an unbounded number of species, this is easy to show.
To see why this is less certain for robust, bounded CRNs, we need a few lemmas.

The natural first question to ask is how can a Turing machine can keep up in real-time with a CRN from the definition of real-time following.
A CRN, after all, is allowed to change all of its species concentrations simultaneously while a Turing machine, following the CRN's memory trajectory and not directly simulating the CRN, must keep all but one (except in the unlikely case where two or more species change memory state at the exact same time) of its species-tracking tapes effectively constant between memory trajectory transitions.

This is not a limitation since a species must linger in a memory state for a minimum length of time.
A CRN with $n$ species and delay $d$ can only experience at most $n$ memory states in every open interval of length $d$ (see Lemma~\ref{CRNs Change State Only So Quickly}).
This is what makes a Turing machine $M$ real-time following a CRN occur in real-time.
It follows that $M$ can compute each of these state changes sequentially while only requiring a constant factor of $\card{S}$ more time in the worst case.

The remaining difficulty is to show that a bounded CRN cannot `cheat' in the sense that a Turing machine would require an infinite alphabet or an infinite number of states or tapes to real-time follow it.
We show this is the case when each memory map has only finite order and the transitions between memory states is memory deterministic.

\begin{theorem}
	\label{CRN Finite Automata are Real-Time Followable}
	Let $N = (S,R)$ be a uniformly $\beta$-bounded CRN, let $f_X \in \mathcal{M}_\beta$ for each $X \in S$ with $ord(f_X) < \infty$, and let $d \in \positives$.
	If $N$ is memory deterministic, then there is a Turing machine which real-time follows $N$ according to $\collection{f_X}{X}{S}$ with delay $d$.
\end{theorem}

\begin{proof}
	Since $N$ is memory deterministic, there is a function $\func[\delta]{\naturals^S}{\naturals^S}$ such that if $N$'s memory trajectory contains memory state $\bm{m} \in \naturals^S$, then $N$'s memory trajectory's next memory state is $\delta(\bm{m})$.
	Define the injective function $\func[\mbox{itoa}]{\naturals}{\bin^*}$ to be the map from natural numbers to their binary expansion with $\mbox{\textbf{itoa}}$ being the equivalent map on $\naturals^S$.
	We now informally describe the Turing machine $M$ which real-time follows $N$.
	
	Suppose we initialize $N$ to $\bm{m_0} \in \naturals^S$.
	$M$ contains $\card{S}$ tapes (indexed by $S$) initialized to $\mbox{\textbf{itoa}}(\bm{m_0})$.
	For each memory state $\bm{m} \in \underset{X \in S}{\times} \mathbb{Z}_{ord \paren*{f_X}}$, $M$ has a state $q_{\bm{m}}$.
	From $q_{\bm{m}}$, $M$ clears each tape and then writes $\mbox{\textbf{itoa}}(\delta(\bm{m}))$ onto its tapes.
	Lastly, $M$ then either transitions to state $q_{\delta(\bm{m})}$ or halts if $N$ has no further memory states in its memory trajectory.
	
	Clearly $M$ follows the strings $N$'s species produces as intended.
	It remains to show that $M$ also satisfies the time requirements.
	Since the order of each memory map is finite, writing each new string onto a tape requires only constant time.
	Moreover, the sequence of steps recording when it finishes each one of these processes is clearly computable from counting the number of transitions in $M$'s state space.
	Finally, Lemma~\ref{CRNs Change State Only So Quickly} tells us that $N$, no matter how efficient, takes \emph{at least} constant time to transition between its memory states regardless of its initialization.
	\qed
\end{proof}

\begin{corollary}
	\label{CRNs which Halt are Real-Time Followable}
	Let $N = (S,R)$ be a uniformly $\beta$-bounded CRN, and let $f_X \in \mathcal{M}_\beta$ for each $X \in S$.
	Fix a delay $d \in \positives$.
	If $N$ is memory deterministic and $N$'s memory trajectory is either finite or there exists a memory state which appears at least twice in it, then there is a Turing machine which real-time follows $N$ according to $\collection{f_X}{X}{S}$ with delay $d$.
\end{corollary}

\begin{proof}
	If $N$'s memory trajectory is finite, then it only uses finitely many memory states.
	Excising the unused states from the memory maps causes each to have only finite order, at which point Theorem~\ref{CRN Finite Automata are Real-Time Followable} applies.
	Similarly, if a memory state appears twice in $N$'s memory trajectory, then because $N$ is memory deterministic, it must necessarily loop through a finite set of memory states forever.
	Because $N$ enters into an infinite loop through a finite set, it uses only finitely many memory states and Theorem~\ref{CRN Finite Automata are Real-Time Followable} again applies.
	\qed
\end{proof}
\section{Discussion}
In this paper, we have shown that \emph{only} the algebraic real numbers are computable by CRNs using exponentially stable equilibria.
Intuitively, this means that every transcendental real number cannot be computed robustly by a CRN in the sense of Definition~\ref{Lyapunov CRN Computability}.
This led us to explore in Section~\ref{sec:mmcrn} what it means for a CRN to compute robustly.
We started from two notions of computation.
First, a CRN can compute a value exactly, which a non-constant CRN can achieve only for a measure zero length of time.
Second, a CRN may compute a value in the limit, which has two problems of its own.
A CRN never precisely achieves a value computed in the limit.
Moreover, we showed earlier in Theorem~\ref{thm:algeqlup} that only the algebraic numbers can be so computed reliably.
Any non-algebraic number, if the CRN is improperly initialized with any epsilon error, cannot be computed in the limit.

These limitations led us to ask what happens when we require a CRN to behave identically for a range of inputs rather than a single set of concentrations. The result was the notion of a memory map, a strangely discrete model of an analog implementation of computation. Arguably, under this model, a CRN's reactions correspond to transitions between states of a Turing machine while species concentrations correspond to tape states.

This ultimately led to Theorem~\ref{CRN Finite Automata are Real-Time Followable}. In the more familiar terminology of the discrete world, it tells us that a robust CRN is no more capable of executing a NFA than a Turing machine is. This is perhaps unsurprising. The main advantage a CRN has to leverage over a Turing machine is in its ability to rewrite its entire tape with a new word of any length. With finite memory, this advantage is lost.

Notice, however, that Theorem~\ref{CRN Finite Automata are Real-Time Followable} says nothing about the \emph{existence} of a robust CRN capable of simulating a NFA. For that, we turn to~\cite{KLINGE2020114} for a CRN with a more restrictive notion of robustness which nonetheless satisfies the definitions we derived here and Theorem~\ref{CRN Finite Automata are Real-Time Followable}. We briefly summarize this below as an illustrative example.

Given a NFA $M$, a CRN $N$ is constructed with two species, $X_q$ and $\overline{X}_q$, for each state $q$ of $M$. These species alternatively take concentrations close to $1$ or $0$ to represent $M$ being in state $q$ or not in state $q$ respectively for $X_q$ and vice versa for $\overline{X}_q$. The appropriate memory map for each of these species would be to map $0$ to an interval around $0$, $1$ to an interval around $1$, and $2$ to everything in-between. The correct delay to choose for this CRN is the length of the clock cycle (which also admits an identical memory map). For the input signal (which, again, admits an identical memory map), we may assume that there is an external CRN generating it which makes the $N$ memory deterministic.

First, note that $N$ is uniformly bounded on all of its valid inputs. Moreover, if we apply Theorem~\ref{CRN Finite Automata are Real-Time Followable} to the CRN described above, we obtain a Turing machine which not only behaves identically but can be transformed back into the same CRN~\cite{KLINGE2020114}. As such, these are truly inverse statements. Moreover, the theorem can be applied to more general cases as well.

Given a Turing machine $M$, Fages \etal construct a GPAC-generable function (easily translated into the CRN world) that simulates $M$ within bounded time and tape space~\cite{Fages2017}. The input parameters for each bound can be adjusted, of course, but once fixed, the resulting simulation permits a single memory map model for all of its input configurations to which Corollary~\ref{CRNs which Halt are Real-Time Followable} applies. In a sense, these, too, are inverse statements.

Now the question becomes where to go from here. It is known that, given an NFA, there is a robust CRN which simulates it in real-time~\cite{KLINGE2020114}. Similarly, we have provided a proof that, given a robust CRN with memory maps of only finite order, there is a Turing machine which real-time follows it. In short, for the regular languages, robust CRNs and Turing machines are fully equivalent models with neither having an advantage over the other. We conjecture that the same is true of an arbitrary robust CRN, that is given a robust CRN with a memory deterministic collection of memory maps, there is a Turing machine which real-time follows it. This, if true, has several important implications.

First, it's known that CRNs and Turing machines can simulate each other with a polynomial-time slowdown~\cite{Pouly2017}. If this conjecture is true, even in a more restricted form, it would eliminate the slowdown from a Turing machine simulating a CRN.

Of perhaps more interest is the Hartmanis-Stearns Conjecture (HSC)~\cite{HartmanisStearns1965}. Both Turing machines and robust CRNs are clearly capable of outputting the digits of a rational number in real-time. For Turing machines, this means writing to some output tape. For CRNs, this \emph{does not} mean outputting a concentration but rather raising a concentration high or low in a memory trajectory in the appropriate sequence. In this manner, assuming our stated conjecture, then if one could construct a robust CRN to output the digits of a nonrational algebraic number, it would also resolve the HSC for Turing machines.

\subsubsection*{Acknowledgments.}
The authors thank anonymous reviewers for useful feedback and suggestions.
This research supported in part by NSF grants 1900716 and 1545028.

\bibliography{nagabib}

\end{document}